\documentclass[onecolumn]{revtex4-2}
\usepackage{graphicx}
\usepackage{amssymb,latexsym,amsmath,amsthm}
\usepackage{braket,bm,bbm}
\usepackage{mathrsfs}
\usepackage{mathtools}
\usepackage[
	bookmarks  = false,
	colorlinks = true,
	linkcolor  = blue,
	urlcolor   = blue,
	citecolor  = blue]{hyperref}
\usepackage[T1]{fontenc}
\usepackage[normalem]{ulem}
\usepackage{bookmark}
\usepackage{natbib}
\usepackage{array}

\DeclareMathOperator{\tr}{Tr}

\newtheorem{thm}{Theorem}

\newtheorem{defn}{Definition}

\newtheorem{prop}{Proposition}

\begin{document}

\title{An Initialization-free Quantum Algorithm for General Abelian Hidden Subgroup Problem}
\author{Sekang Kwon}
\affiliation{Department of Mathematics, Kyung Hee University, Seoul 02447, Republic of Korea}
\author{Jeong San Kim}
\email{freddie1@khu.ac.kr}
\affiliation{Department of Applied Mathematics and Institute of Natural Sciences, Kyung Hee University, Yongin 17104, Republic of Korea}

\setlength{\parindent}{18pt}

%%%%%%%%%%%%%%%%%%%%%%%%%%%%%%%
%%%%   Abstract     %%%%%%%%%%%
%%%%%%%%%%%%%%%%%%%%%%%%%%%%%%%
\begin{abstract}
Hidden Subgroup Problem (HSP) seeks to identify an unknown subgroup $H$ of a group $G$ for a given injective function $f$ defined on cosets of $H$. Here, we present an initialization-free quantum algorithm for solving HSP in the case where $G$ is a finite abelian group. Our algorithm can adopt an arbitrary unknown mixed state as the auxiliary register and removes the need for initialization while preserving computational cost comparable to existing methods. Our algorithm also restores the state of the auxiliary register to its original form after completing the computations. Since the recovered state can be utilized for other operations, a single preparation of the auxiliary register in an arbitrarily unknown mixed state is sufficient to execute the iterative procedure in solving hidden subgroup problems. This approach provides a promising direction for improving quantum algorithm efficiency by reducing operational time of initialization.
\end{abstract}
\maketitle
	
%%%%%%%%%%%%%%%%%%%%%%%%%%%%%%%
%%%%   Introduction     %%%%%%%
%%%%%%%%%%%%%%%%%%%%%%%%%%%%%%%
\section{Introduction}\label{intro}
	
Quantum computing introduces a novel paradigm for information processing, leveraging principles of quantum mechanics such as superposition and entanglement. Unlike classical computation, which uses bits as the fundamental units of information, quantum computation employs qubits, enabling the parallel exploration of multiple computational pathways. This unique capability allows quantum algorithms to solve certain computational problems exponentially faster than their classical counterparts. Notable examples include Simon's problem~\cite{Sim97} and the period-finding problem~\cite{Sho97}, which demonstrate the transformative potential of quantum computation.

Both Simon's problem and the period-finding problem can be understood as specific instances of a broader class of problems known as the \emph{Hidden Subgroup Problem (HSP)}. This framework, rooted in group theory, provides a unified description of such problems. Each instance of the HSP is characterized by the group under consideration. For example, Simon's problem is associated with the Boolean group $(\mathbb{Z}_2^n, \oplus)$, while the period-finding problem corresponds to the cyclic group $(\mathbb{Z}_N, +)$. Beyond these examples, there are diverse instances of the HSP arising from both abelian and non-abelian groups~\cite{BCvD05,CKL06,DJ92,Gav04,Hal05,Kup05}.
	
For the abelian HSP (AHSP), where the underlying group $G$ is abelian, the problem can be solved in polynomial time with respect to the logarithm of the group order, $\log{|G|}$~\cite{Lom04}. This success positions the HSP as a cornerstone in the development of quantum algorithms illustrating their advantage over classical algorithms. 
However, implementing quantum algorithms on real quantum computers introduces several practical challenges, including noise, gate imperfections, and the cost of state preparation. In particular, the initialization of quantum registers into a prescribed pure state can represent a non-negligible overhead, especially in contexts where quantum subroutines are executed repeatedly or intermediate quantum states must be reused. Under such conditions, state initialization may become a substantial experimental bottleneck.

To mitigate this issue, \emph{initialization-free quantum algorithms} have been proposed~\cite{CKL01, CKL05}. These algorithms allow the use of an arbitrary unknown mixed state as the auxiliary register, while ensuring that the state is restored after computation. By virtue of these features, initialization-free quantum algorithms can avoid auxiliary register initialization and utilize intermediate quantum states generated by other computations, while maintaining computational efficiency comparable to conventional quantum algorithms. This can also improve the space efficiency in the scenarios where intermediate quantum states are reused as auxiliary registers. However, existing works on initialization-free algorithms have been limited to specific problems such as the Deutsch-Jozsa problem, Simon's problem, and period-finding~\cite{CKL01, CKL05}. 

Here, we present an initialization-free quantum algorithm for the AHSP, extending the framework to all finite abelian groups. Our algorithm can adopt an arbitrary unknown mixed state as the auxiliary register and removes the need for initialization of the auxiliary register while preserving computational cost comparable to existing methods. Our algorithm also restores the state of the auxiliary register to its original form after completing the computations. Since many important quantum algorithms that achieve exponential speed-ups, including Simon's and Shor's algorithms, can be formulated within the framework of the AHSP, our result broadens the applicability of initialization-free techniques.

This paper is organized as follows. In Section~\ref{pre}, we introduce some definitions and propositions as preliminaries. In Section~\ref{sec:stan}, we review the standard quantum algorithm for AHSP. In Section~\ref{sec:ifqa}, we present our initialization-free algorithm for AHSP. In Section~\ref{sec:conclusion}, we conclude by summarizing our results and some future works.

%%%%%%%%%%%%%%%%%%%%%%%%%%%%%%%
%%%%   Preliminaries    %%%%%%%
%%%%%%%%%%%%%%%%%%%%%%%%%%%%%%%
\section{Preliminaries}\label{pre}
\subsection{Hidden Subgroup Problem}\label{sec:hidden subgroup problem}
In this section, we first recall the definition of the HSP and provide some basic concepts of the representation theory for finite abelian groups.
\begin{defn}{\textnormal{(HSP)}}\label{defn:HSP} \cite{NC00} 
Let $G$ be a group and $H$ be a subgroup of $G$. Consider a function $f$ from $G$ to a finite set $Y$ satisfying the following condition,
    \begin{equation}\label{separating} 
        \quad f(x)=f(y)\quad\text{if and only if}\quad xH=yH
    \end{equation}
    for all $x,y\in G$. HSP is to identify a generating set of $H$ using evaluations of $f$. 
\end{defn}
This definition of HSP provides the formal framework for many algebraic and number theoretic problems including Simon's problem and integer factorization: Simon's problem can be considered as the HSP with the Boolean group $(\mathbb{Z}_2^n, \oplus)$, and  the integer-factorization can be characterized as a period-finding problem, which is the HSP on the cyclic group $(\mathbb{Z}_N, +)$.
Classically, HSP is believed to be hard to solve in general. However, for AHSP, where $G$ is an abelian group, there exists a polynomial-time quantum algorithm with respect to $\log{|G|}$.

%\subsection{Character theory for abelian groups}\label{sec:character theory for}

Now we review some representation-theoretic notions for finite abelian groups to describe the quantum Fourier transform over abelian groups and the structure of coset states.
\begin{defn}{\textnormal{(Characters and Dual Group)}}\label{defn:character}
	Let $G$ be a group. A \emph{character} of $G$ is a homomorphism $\chi:G\to \mathbb{C}^*$, where $ \mathbb{C}^*$ is the multiplicative group of non-zero complex numbers. If $G$ is abelian, then the set of all characters of $G$ forms a group under pointwise multiplication, which we call the \emph{dual group} of $G$ denoted by $\hat{G}$.
\end{defn}

For a finite abelian group $G$, the dual group $\hat{G}$ is isomorphic to $G$~\cite{Note1}. Hence each character in $\hat{G}$ can be labeled by an element $g\in G$ denoted by $\chi_g$ and we will abuse notation by writing $g$ instead of $\chi_g$ when there is no confusion.
The \emph{annihilator} of $H$ is a subgroup of $\hat{G}$, which is useful when analyzing the quantum Fourier transform over abelian groups.

\begin{defn}{\textnormal{(Annihilator)}}\label{defn:annihilator}
	Let $G$ be an abelian group and $H$ be a subgroup of $G$. The \emph{annihilator} of $H$ is the subgroup $H^{\perp}\subseteq\hat{G}$ defined as
	\begin{equation}\label{H_perp}
		H^{\perp} = \{\chi_g \in \hat{G} \mid \chi_g(h) = 1 \text{ for all } h \in H\}.
	\end{equation}
\end{defn}
In other words, the annihilator of a subgroup $H$ consists of all characters that are trivial on the subgroup $H$. 
%In other words,  $\res_H^G\chi_g=\chi_e$ where $\res_H^G\chi_g$ is the restriction of $\chi_g$ to $H$ %and $e$ is the identity element of $H$. 
This notion naturally appears when evaluating the sum of character values over a subgroup. In particular, the sum of character values exhibits a strong cancellation property unless the character is trivial on the subgroup. This behavior follows from the orthogonality relation of characters.

\begin{prop}{\textnormal{(Orthogonality of characters)}}\label{prop:orthogonality of characters}
	Let $G$ be a finite abelian group and $\hat{G}$ be its dual group. For any $\chi_g, \chi_{g'} \in \hat{G}$, we have
	\begin{equation}\label{orthogonality}
		\frac{1}{|G|} \sum_{x \in G} \overline{\chi_g(x)} \chi_{g'}(x) = \delta_{g,g'},
	\end{equation}
	where $\delta_{g,g'}$ is the Kronecker delta function, which is equal to 1 if $g=g'$ and 0 otherwise.
\end{prop}
As a direct consequence, the sum of character values over a subgroup $H$ is equal to $|H|$ if the character is trivial on $H$ and $0$ otherwise. 
	
\subsection{The quantum Fourier transform over abelian groups and the coset state}\label{sec:the quantum fourier transform}

The Quantum Fourier Transform (QFT) is a core primitive in quantum algorithms for HSP. In this subsection, we introduce the QFT over finite abelian groups and describe its action on coset states.

\begin{defn}{\textnormal{(QFT over finite abelian groups)}}\label{defn:qft over abelian}
	Let $G$ be a finite abelian group and $\hat{G}$ be its dual group. The QFT over $G$, denoted by $\mathcal{F}_G$, is the unitary operator on $\mathbb{C}^{|G|}$ defined as
	\begin{equation}
		\mathcal{F}_G\ket{x} = \frac{1}{\sqrt{|G|}} \sum_{g \in \hat{G}} \chi_g(x) \ket{g},
	\end{equation}
	for each $x\in G$.
\end{defn}

Let $G$ be a finite ableian group and $H$ be a subgroup of $G$. For a coset of the subgroup $H$, its coset state is defined as the uniform superposition of all the elements in the coset;  
\begin{defn}\textnormal{(Coset state)}\label{defn:coset state}
	Let $G$ be a finite abelian group and $H$ be a subgroup of $G$. For a representative $r\in G$, the \emph{coset state} associated with the coset $r+H$ is defined as
	\begin{equation}
		\ket{r+H} = \frac{1}{\sqrt{|H|}} \sum_{h \in H} \ket{r+h}.
	\end{equation}
\end{defn}

For a finite abelian group $G$ with a subgroup $H$, let us consider the action of the QFT on a coset state; applying $\mathcal{F}_G$ to $\ket{r+H}$ gives
\begin{align}\label{eq:qft_coset_state}	
	\mathcal{F}_G\ket{r+H} & = \frac{1}{\sqrt{|H|}} \sum_{h \in H} \mathcal{F}_G\ket{r+h}\nonumber \\
	& = \frac{1}{\sqrt{|H|}} \sum_{h \in H} \frac{1}{\sqrt{|G|}} \sum_{g \in \hat{G}} \chi_g(r+h) \ket{g}\nonumber \\
	& = \frac{1}{\sqrt{|H||G|}} \sum_{g \in \hat{G}} \chi_g(r) \left( \sum_{h \in H} \chi_g(h) \right) \ket{g}. 
\end{align}

By the orthogonality of characters, we have
\begin{equation}\label{sum_otho}
	\sum_{h \in H} \chi_g(h) = \begin{cases}
		|H| & \text{if } g \in H^{\perp}, \\
		0 & \text{otherwise}.
	\end{cases}
\end{equation}

From Eq.~(\ref{eq:qft_coset_state}) together with Eq.~(\ref{sum_otho}), we have 
\begin{equation}\label{QFTco}
	\mathcal{F}_G\ket{r+H}=\sqrt{\frac{|H|}{|G|}}\sum_{g\in H^{\perp}}\chi_g(r)\ket{g}.
\end{equation}
In other words, applying the QFT to a coset state produces a superposition over characters in the annihilator $H^{\perp}$. This allows us to extract information about the subgroup $H$ by performing QFT to coset states, which is a key step in efficient quantum algorithms for HSP.

%%%%%%%%%%%%%%%%%%%%%%%%%%%%%%%%%%%%%%%%%%%%%%%%%%%%%%%%
%%%%   The standard quantum algorithm for AHSP   %%%%%%%
%%%%%%%%%%%%%%%%%%%%%%%%%%%%%%%%%%%%%%%%%%%%%%%%%%%%%%%%

\section{The standard quantum algorithm for AHSP}\label{sec:stan}
Before we provide an initialization-free quantum algorithm for AHSP in the next section, we review the standard polynomial-time quantum algorithm for AHSP \cite{Wol19}. For a finite abelian group $G$ and its subgroup $H$ hidden by the function $f$ satisfying the condition \eqref{separating}, the algorithm finds a generating set of $H$ in polynomial time with respect to $\log{|G|}$.
 
To process the function evaluation, the algorithm employs a quantum oracle $\mathcal{U}_f$ defined as 
\begin{equation}\label{eq:U_f}
\mathcal{U}_f\ket{x}_A\ket{y}_B=\ket{x}_A\ket{y+f(x)}_B,
\end{equation}
where the subscript $A$ denotes the main register to deal with the group elements $x$ in $G$, and $B$ is the auxiliary register for storing the function values of $f$. Here we note that the binary operation $+$ in the auxiliary register can be justified in the context of defining the quantum oracle $\mathcal{U}_f$, although the set $Y$ does not inherently possess any algebraic structure.
Moreover, to support the initialization-free quantum algorithm for AHSP in Section \ref{sec:ifqa}, we assume that $f$ is {\em surjective} on $Y$ throughout this paper, which allows us to regard $Y$ as a finite abelian group isomorphic to $G/H$.

The standard algorithm for AHSP is as follows: 
\begin{enumerate}
	\item Prepare the initial state in $\ket{0}_A\ket{0}_B$.
	\item \label{item:quantum Fourier transform} Apply $\mathcal{F}_G$ on the main register $A$ to yield 
	\begin{equation}
	\frac{1}{\sqrt{|G|}}\sum_{x\in G}\ket{x}_A\ket{0}_B.	
	\end{equation}
	%$\frac{1}{\sqrt{|G|}}\sum_{x\in G}\ket{x}_A\ket{0}_B$.
	\item \label{item:function evaluation} Evaluate function values of $f$ in the register $B$ by querying the quantum oracle $\mathcal{U}_f$ to obtain 
	\begin{equation}\label{eq:function evalutation}
		\frac{1}{\sqrt{|G|}}\sum_{x\in G}\ket{x}_A\ket{f(x)}_B=\sqrt{\frac{|H|}{|G|}}\sum_{r\in R}\ket{r+H}_A\ket{f(r)}_B,	
	\end{equation}
	where $R$ is a set of representatives of the cosets in $G/H$.
	\item \label{item:second fourier transform} Apply $\mathcal{F}_G$ on $A$ again to obtain 
	\begin{equation}\label{eq:stan final}
		\frac{|H|}{|G|}\sum_{r\in R}\sum_{h\in H^{\perp}}\chi_{h}(r)\ket{h}_A\ket{f(r)}_B=:\ket{\psi_{std}}_{AB},
	\end{equation}
	where $R$ is a set of representatives of the cosets in $G/H$.
	\item Measure the register $A$ in the computational basis.
\end{enumerate}
From Step \ref{item:function evaluation} to \ref{item:second fourier transform}, we expressed register A in Eq.~\eqref{eq:function evalutation} into coset states and used the property of the coset state discussed as Eq.~(\ref{QFTco}) of Section~\ref{sec:the quantum fourier transform}.
 
By measuring the register $A$ in the computational basis, we obtain a random sample $g$ of $H^{\perp}$ with probability $|H|/|G|$ that comes from the following equalities:
\begin{align}\label{stan:prob}
	\Pr(g)=&\tr\left[\left(\ket{g}_A\negthinspace\bra{g}\otimes I_B\right)\ket{\psi_{std}}_{AB}\negthinspace\bra{\psi_{std}}\right]\nonumber\\
	=&\frac{|H|^2}{|G|^2}\sum_{r,r'\in R}\sum_{h,h'\in H^{\perp}}\overline{\chi_{h'}(r')}\chi_h(r)\braket{h'|g}_A\negthinspace\braket{g|h}\braket{f(r')|f(r)}_B\nonumber\\
	=&\frac{|H|^2}{|G|^2}\sum_{r\in R}1\nonumber\\
	=&\frac{|H|}{|G|}.
\end{align}

By iterating this process $\mathcal{O}(\log|G|)$ times, we can find a generating set of $H^{\perp}$ with high probability. Then we can efficiently restore a generating set of $H$ from $H^{\perp}$ by solving a system of linear equations \cite{Lom04}. Because we query once per each iteration, the number of queries equals to the number of iteration. In addition, from the fact that QFT can be achieved in $\mathcal{O}(\log^2{|G|})$ operations \cite{NC00}, the number of operations is $\mathcal{O}(\log^3|G|)$. The following proposition summarizes the standard quantum algorithm for AHSP.
\begin{prop}\label{prop:stan} 
	Let $G$ be a finite abelian group and $H$ be a subgroup of $G$. Adopting a quantum oracle $U_f$ defined in Eq.~\eqref{eq:U_f}, the standard quantum algorithm for AHSP solves the problem with $\mathcal{O}(\log{|G|})$ oracle queries and $\mathcal{O}(\log^3|G|)$ operations.
\end{prop}

\begin{figure*}[t!]
	\centering
	\includegraphics[width=0.75\textwidth]{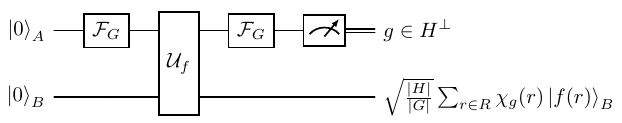}
	\caption{The quantum circuit of standard quantum algorithm for AHSP. $\mathcal{F}_G$ denotes the QFT on the group $G$ and $R$ is a set of representatives of $G/H$. The post-measurement state of the auxiliary register $B$ is a superposition over the whole computational basis. One has to restore the initial state of $B$ to perform another round of the algorithm.}\label{fig:stan}
\end{figure*}

Here, we would like to pay attention to the post-measurement state of the auxiliary register $B$ after the measurement of the register $A$. When we measure the register $A$ and obtain the outcome $g$, the auxiliary register $B$ will be in the state  
\begin{equation}\label{stan:outcome}
	\sqrt{\frac{|H|}{|G|}}\sum_{r\in R}\chi_{g}(r)\ket{f(r)}_B,
\end{equation}
which we illustrate in Figure~\ref{fig:stan}.

As one has to iterate the algorithm to get a solution of an AHSP, the initialization problem arises here because the auxiliary register $B$ is random state of the form Eq.~\eqref{stan:outcome} depending on the measurement outcome $g$. Thus one has to restore the initial state of $B$ to perform another round of the algorithm. In the following section, we present an algorithm that does not require the initialization of the auxiliary register and restores its initial state at the end of computation.

%%%%%%%%%%%%%%%%%%%%%%%%%%%%%%%%%%%%%%%%%%%%%%%%%%%%%%%%%%%%%%%%
%%%%   Initialization-free Quantum Algorithm for AHSP    %%%%%%%
%%%%%%%%%%%%%%%%%%%%%%%%%%%%%%%%%%%%%%%%%%%%%%%%%%%%%%%%%%%%%%%%

\section{Initialization-free Quantum Algorithm for Abelian Hidden Subgroup Problem}\label{sec:ifqa}

In this section, we present an initialization-free quantum algorithm for AHSP utilizing an arbitrary unknown mixed state as the auxiliary register. Our algorithm restores its initial state at the end of the computation, while preserving computational cost comparable to the standard one in the previous section. 

 we define a unitary operator $\mathcal{S}_z$ composed of QFTs and a translation, which is used in our algorithm.
\begin{defn}{\textnormal{(The unitary operator $\mathcal{S}_z$)}}\label{defn:S_z}
	Let $Y$ be a set having $M$ elements. For a given $z\in Y$, a unitary operator $\mathcal{S}_z$ on $\mathbb{C}^{M}$ is defined as
	\begin{equation}\label{S_z}
		\mathcal{S}_z=\mathcal{F}_M\mathcal{T}_{z}^{\dagger}\mathcal{F}_M,
	\end{equation}
	where $\mathcal{F}_M$ is the QFT over a cyclic group of order $M$ and $\mathcal{T}_z$ is the translation operator that shifts $\ket{y}$ to $\ket{y+z}$ for every $y\in Y$.
\end{defn}

For each $y\in Y$, the action of $\mathcal{S}_z$ on $\ket{y}$ is given by
\begin{equation}\label{S_z_action}
	\mathcal{S}_z\ket{y}=\omega_{M}^{yz}\ket{-y},
\end{equation}
where $\omega_M=e^{2\pi i/M}$. We will use $\mathcal{S}_z$ to encode the function values of $f$ as phases in the first register and to return the state of the auxiliary register to its initial state.

The initialization-free quantum algorithm for AHSP adopts an unknown mixed state $\rho_B=\sum_{y\in Y}\lambda_y\ket{y}_B\negthinspace\bra{y}$ as the auxiliary register, hence the initial state of our algorithm is considered to be prepared in $\ket{0}_A\ket{y}_B$ with probability $\lambda_y$. The algorithm proceeds as follows:
\begin{enumerate}
	\item Prepare the initial state in $\ket{0}_A\ket{y}_B$, with probability $\lambda_y$ where $y\in Y$.
	\item %Similarly with Step \ref{item:quantum Fourier transform} and \ref{item:function evaluation} in Sec. \ref{sec:stan}, 
	Apply $\mathcal{F}_G$ on the register $A$ followed by querying to the quantum oracle $\mathcal{U}_f$ to obtain
	\begin{equation}\label{eq:function evaluation y}
		\frac{1}{\sqrt{|G|}}\sum_{x\in G}\ket{x}_A\ket{y+f(x)}_B=\sqrt{\frac{|H|}{|G|}}\sum_{r\in R}\ket{r+H}_A\ket{y+f(r)}_B.
	\end{equation}
	\item For a randomly chosen $z\in Y$, apply the unitary operator $\mathcal{S}_z$ on the register $B$ to yield
	\begin{equation}\label{eq:first S_z}
		\sqrt{\frac{|H|}{|G|}}\sum_{r\in R}\omega_M^{z(y+f(r))}\ket{r+H}_A\ket{-\left(y+f(r)\right)}_B.
	\end{equation}
	\item Query to $\mathcal{U}_f$ and applying $\mathcal{S}_z$ on the register $B$ again to obtain
	\begin{equation}\label{eq:second S_z}
		\sqrt{\frac{|H|}{|G|}}\sum_{r\in R}\omega_M^{zf(r)}\ket{r+H}_A\ket{y}_B.
	\end{equation}
	\item Apply $\mathcal{F}_G$ on the register $A$ to obtain
	\begin{equation}\label{eq:if final}
		\left(\frac{|H|}{|G|}\sum_{r\in R}\sum_{h\in H^{\perp}}\omega_M^{zf(r)}\chi_h(r)\ket{h}_A\right)\ket{y}_B=:\ket{\psi_{if}^{(z)}}_A\ket{y}_B.
	\end{equation}
	\item Measure the register $A$ in the computational basis.
\end{enumerate}

In Eq.~\eqref{eq:if final}, coefficients do not depend on $y$, hence the left-hand side of Eq.~\eqref{eq:if final} can be written in product form as the right-hand side. Because every basis state $\ket{y}_B$ of $\rho_B$ is recovered at the end of computation, $\rho_B$ is also recovered by virtue of linearity of unitary operations.

By measuring $\ket{\psi_{if}^{(z)}}_A$, we can sample an element $g\in H^{\perp}$. Note that the final state $\ket{\psi_{if}^{(z)}}_A$ depend on the random choice of $z$; for a given $z \in Y$,
%. 
%Thus, we consider the expected conditional probability $\mathbb{E}_{z\in Y}[\Pr(g\mid z)]$ over $z$. We first calculate 
the conditional probability $\Pr(g\mid z)$ of obtaining $g$ is
% when an arbitrary $z\in Y$ is given,
\begin{align}\label{eq:if prob_z}
	\Pr(g\mid z)=&\tr\left[\ket{g}_A\negthinspace\braket{g | \psi_{if}^{(z)}}_A\negthinspace\bra{\psi_{if}^{(z)}}\right]\nonumber\\
	=&\frac{|H|^2}{|G|^2}\sum_{r,r'\in R}\omega_M^{z(f(r)-f(r'))}\chi_g(r-r').
\end{align}
Because $z$ is chosen at uniformly random, the expected probability $\mathbb{E}_{z\in Y}[\Pr(g\mid z)]$ of obtaining $g$ over all $z \in Y$ is
\begin{align}\label{eq:if expected prob}
	\mathbb{E}_{z\in Y}[\Pr(g\mid z)]=&\frac{1}{|Y|}\sum_{z\in Y}\Pr(g\mid z)\nonumber\\
	=&\frac{|H|^2}{|G|^2|Y|}\sum_{r,r'\in R}\sum_{z\in Y}\omega_M^{z(f(r)-f(r'))}\chi_g(r-r')\nonumber\\
	=&\frac{|H|}{|G|}
\end{align}
where the last equality is from 
%the fact that the sum of $\omega_M^{z(f(r)-f(r'))}$ over $z$ is zero unless $f(r)=f(r')$ and $|Y|=|G|/|H|$. 
\begin{equation}\label{sum_ff2}
	\sum_{z \in Y} \omega_M^{z(f(r)-f(r'))} = \begin{cases}
		|Y|=|G|/|H| & \text{if}~f(r)=f(r'), \\
		0 & \text{if}~ f(r)\neq f(r').
	\end{cases}
\end{equation}
In Figure~\ref{fig:ifqa}, we illustrate our initialization-free algorithm with a circuit diagram. 

\begin{figure*}[t!]
	\centering
	\includegraphics[width=0.75\textwidth]{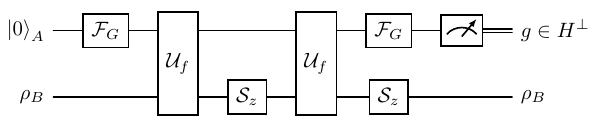}
	\caption{The initialization-free quantum circuit for AHSP. $\mathcal{F}_G$ denotes the QFT on the group $G$. This algorithm employs an arbitrary unknown mixed state $\rho_B$ as the initial state of the register $B$. When one measures the register $A$, the same probability distribution to the standard quantum algorithm in Sec.~\ref{sec:stan} is given. In addition, $\rho_B$ is restored to the initial state at the end of the computation.}
	\label{fig:ifqa}
\end{figure*}

As the expected probability of obtaining $g\in H^{\perp}$ in Eq.~(\ref{eq:if expected prob}) is same as that of the standard algorithm, Eq.~(\ref{stan:prob}) in Section~\ref{sec:stan}, we need same number of iterations to find a generating set of $H^{\perp}$. In addition, the implementation of $\mathcal{S}_z$ requires $\mathcal{O}(\log^2|Y|)$ operations \cite{CKL00}, which is negligible due to the the surjectivity of $f:G\to \protect Y$ therefore $|Y|\leq|G|$. In other words, our algorithm requires $\mathcal{O}(\log{|G|})$ oracle queries and $\mathcal{O}(\log^3|G|)$ operations, which is the same as the standard quantum algorithm for AHSP. We finish this section with the following theorem that summarizes our results.

\begin{thm}{\textnormal{(The initialization-free quantum algorithm for AHSP)}}\label{thm:ifqa} 
	Let $G$ be a finite abelian group and $H$ be a subgroup of $G$ hidden by a surjective oracle function $f:G\to \protect Y$. The initialization-free quantum algorithm for AHSP requires same resources as the standard quantum algorithm for AHSP. That is, it requires $\mathcal{O}(\log{|G|})$ oracle queries and $\mathcal{O}(\log^3|G|)$ operations.
\end{thm}

Since the auxiliary register is restored to its initial state, the algorithm enables one to utilize intermediate quantum states as the auxiliary register, which can reduce space requirements.

%%============%%
%% Conclusion %%
%%============%%
\section{Conclusion}\label{sec:conclusion}

In this work, we presented an initialization-free quantum algorithm for AHSP. 
Our algorithm incorporates two applications of the additional operation $S_z$ and one additional oracle query into each iteration of the standard quantum algorithm for AHSP. 
These operations encode the function values of $f$ into phases and uncompute the auxiliary register, allowing the use of an arbitrary unknown mixed state as the auxiliary register while restoring it to its original state after the computation. 
Although the algorithm requires additional operations, it achieves the same query and time complexity as the standard quantum algorithm, requiring $O(\log |G|)$ oracle queries and $O(\log^3 |G|)$ elementary operations.

Unlike conventional approaches, our algorithm does not require the auxiliary register to be initialized and it can adopt an arbitrary unknown mixed state. These features 
make our algorithm more practical for implementation on current quantum platforms. In addition, the ability to reuse intermediate quantum states can further improve space efficiency. Since many quantum algorithms providing exponential speed-up are reduced to AHSP, our method can be applied to a broad class of computational problems.

%%%%%%%%%%%%%%%%%%%%%%
%  Acknowledgments   %
%%%%%%%%%%%%%%%%%%%%%%
\section*{Acknowledgments}
This work was supported by Korea Research Institute for defense Technology planning and advancement (KRIT) grant funded by Defense Acquisition Program Administration(DAPA)(KRIT-CT-23–031), and the Institute for Information \& Communications Technology Planning \& Evaluation(IITP) grant funded by the Korean government(MSIP)(Grant No. RS-2025-02304540). JSK was supported by Creation of the Quantum Information Science R\&D Ecosystem(Grant No. 2022M3H3A106307411) through the National Research Foundation of Korea(NRF) funded by the Korean government(Ministry of Science and ICT).
%%%%%%%%%%%%%%%%%%%%%%

%%%%%%%%%%%%%%%%%%%%%%
%     References     %
%%%%%%%%%%%%%%%%%%%%%%

%%%%%%%%%%%%%%%%%%%%%%
\end{document}